\documentclass[runningheads]{llncs}

\usepackage[T1]{fontenc}
\usepackage{amsmath,amssymb}
\usepackage{microtype}
\usepackage{url}

\newcommand{\FMMP}{\textsc{FMMP}}
\newcommand{\XP}{\textsc{XP}}
\newcommand{\CountFMMP}{\#\textsc{FMMP}}
\newcommand{\rank}{\operatorname{rank}}
\spnewtheorem{observation}[theorem]{Observation}{\bfseries}{\itshape}

\begin{document}

\title{Structural Complexity of Matching-Match: Dense and Sparse Graphs}
\titlerunning{Structural Complexity of Matching-Match}

\author{Ilie Dumitru\inst{1} \and Adrian Micl\u au\c s\inst{1} \and Alexandru Popa\inst{1,2}}
\authorrunning{Dumitru, Micl\u au\c s, and Popa}

\institute{Department of Computer Science, University of Bucharest, Str. Academiei 14, Bucharest, 010014, Romania
\and
National Institute for Research and Development in Informatics, Bulevardul Mare\c{s}al Alexandru Averescu 8--10, Bucharest, 011555, Romania}

\maketitle

\begin{abstract}
The Matching-Match puzzle asks whether the vertices of a fixed graph can be colored so that the multiset of color pairs induced by its edges is exactly a prescribed multiset. We study how the complexity of this realization problem depends on the host graph. On the dense side, we give a polynomial-time algorithm for complete $k$-partite graphs for every fixed number $k$ of parts, with arbitrary precoloring and an arbitrary number of colors. We prove a sharp complement-degree threshold: the problem is polynomial-time solvable when $\Delta(\overline G)\le1$, but NP-complete on completely uncolored graphs already when $\Delta(\overline G)=2$. This yields a dichotomy for uniform complete multipartite graphs, and connected diameter two already suffices for NP-completeness. We also prove W[1]-hardness on cographs parameterized by the number of colors.

On the sparse side, completely uncolored paths and cycles admit a linear-time characterization by Euler trails and circuits, while counting feasible colorings is $\#P$-complete on both classes. Counting is nevertheless polynomial-time solvable on stars and complete graphs, even with arbitrary precoloring. A decomposition-transfer theorem yields NP-completeness already at tree-depth two. A separate path-decomposition reduction gives a maximum-degree threshold between one and two for completely uncolored disconnected host graphs with unrestrictedly many colors. Components with at most two edges are tractable, while a disjoint union of $P_4$'s is NP-complete. Finally, precoloring restores tractability in several cases: star forests are polynomial when every center is precolored, and two broad precoloring regimes on length-two spiders are polynomial even when the number of colors is unbounded.

\end{abstract}

\section{Introduction}

\subsection*{Motivation}

Complex networks are commonly described through small recurring patterns of interactions, known as \emph{network motifs}~\cite{MiloEtAl2002}. Such patterns have been used as local structural signatures in biological, technological, ecological, and information-processing networks. When network vertices represent entities with different types, functions, or roles, this viewpoint can be refined through \emph{colored motifs}, where colors retain information carried by the vertices~\cite{AdamiEtAl2011}.

Most motif-based approaches start from an already labeled network and ask which local patterns occur in it. This suggests an inverse realization question: given an unlabeled network topology and a prescribed multiset of colored pairwise interactions, can the vertices be labeled so that exactly those interactions are realized? We study this question through the \emph{Matching-Match puzzle} introduced by Iburi and Uehara~\cite{IburiUehara2024}. The puzzle interpretation provides a concrete origin for the problem, while the interaction-profile viewpoint connects it to structural questions about labeled networks.

\subsection*{Informal Problem Definition}

An instance consists of a graph $G=(V,E)$, which we call the \emph{host graph}, together with exactly $|E|$ colored sticks. Each stick has one color at each endpoint and hence specifies an unordered pair of colors. Some vertices of $G$ may already be precolored. We ask whether the sticks can be assigned bijectively to the graph edges and the remaining vertices can be colored so that all stick endpoints meeting at the same graph vertex have one common color and every precolored vertex keeps its prescribed color.

Equivalently, a final vertex coloring of $G$ induces an unordered color pair on every graph edge. The instance is feasible exactly when the multiset of these induced pairs equals the prescribed multiset of stick types. We call this decision problem \emph{Feasible Matching-Match} and abbreviate it by \FMMP{}. Thus, every use of \FMMP{} below refers to this feasibility problem; the formal notation is fixed in Section~\ref{sec:prelim}.

\subsection*{Previous and Related Work}

Iburi and Uehara introduced Matching-Match and initiated its complexity study~\cite{IburiUehara2024}. They proved NP-completeness on paths, cycles, and spiders, and gave polynomial-time algorithms for complete graphs, stars, and completely uncolored spiders whose legs have length at most two. Their work also identified the number of colors, precoloring on short spiders, and bounded leg length as natural directions for further study.

Dumitru, Micl\u au\c s, and Popa prove that two colors already suffice for NP-completeness on general graphs and develop parameterized results on spiders~\cite{DumitruMiclausPopaSpiders}. They also study the optimization variant \textsc{MaxMMP}, establishing approximation guarantees and exact algorithms for fixed numbers of colors on trees and cographs~\cite{DumitruMiclausPopaMax}. The present paper focuses instead on structural graph classes, sharp tractability boundaries, and counting complexity.

The interaction-profile viewpoint should not be confused with the classical \textsc{Graph Motif} problem, where a vertex-colored graph and a multiset of colors are given and one asks for a connected vertex set realizing that color multiset~\cite{FellowsEtAl2011}. In Matching-Match the vertex coloring itself is part of the solution, and the prescription concerns the color pairs induced by all graph edges simultaneously.

\subsection*{Our Results}

We establish several structural boundaries for \FMMP{}.

\paragraph{Dense graph classes.}
For every fixed number of parts, \FMMP{} is polynomial-time solvable on complete multipartite graphs, even with arbitrary precoloring and an unrestricted number of colors. We obtain a sharp threshold in the maximum degree of the complement and a dichotomy for complete multipartite graphs with equal part size. We also show that NP-completeness already occurs on connected completely uncolored graphs of diameter two.

\paragraph{Parameterized complexity.}
On cographs, \FMMP{} is W[1]-hard when parameterized by the number of colors. This complements the known polynomial-time algorithms for every fixed number of colors.

\paragraph{Sparse graphs and counting.}
Completely uncolored paths and cycles are solvable in linear time. Their counting versions are nevertheless $\#P$-complete. In contrast, feasible colorings can be counted in polynomial time on stars and complete graphs, even with arbitrary precoloring.

\paragraph{Disconnected graphs and short spiders.}
For completely uncolored disconnected host graphs, we obtain hardness from graph-decomposition problems, including hardness at tree-depth two, at maximum degree two, and when every component is a three-edge path. We also identify polynomial cases for star forests with precolored centers and for two broad precoloring patterns on spiders whose legs have length at most two.

\section{Preliminaries}\label{sec:prelim}

Throughout the paper, graphs are finite, simple, and undirected. For a graph $G=(V,E)$, write $n=|V|$ and $m=|E|$, and let $\Delta(G)$ denote its maximum degree. The complement of $G$ is denoted by $\overline G$. The distance between two vertices is the length, in edges, of a shortest path between them, and $\operatorname{diam}(G)$ denotes the diameter of a connected graph. The path on $\ell$ vertices is denoted by $P_\ell$.

The \emph{disjoint union} of graphs is denoted by $\dot\cup$. The \emph{join} $G_1\vee G_2$ is obtained from their disjoint union by adding all edges between $V(G_1)$ and $V(G_2)$. A complete $k$-partite graph with nonempty parts of sizes $n_1,\ldots,n_k$ is denoted by $K_{n_1,\ldots,n_k}$. A \emph{cograph} is a graph obtained from single vertices by repeated disjoint union and join; equivalently, it contains no induced $P_4$. A \emph{block graph} is a graph in which every maximal biconnected subgraph is a clique, and a \emph{linear forest} is a disjoint union of paths.

A \emph{star} is a tree with one center adjacent to every other vertex, and a \emph{star forest} is a disjoint union of stars. A \emph{spider} is a tree with at most one vertex of degree greater than two. In a non-path spider this unique high-degree vertex is its \emph{body}; each maximal path from the body to a leaf is a \emph{leg}, and the length of a leg is its number of edges. In the spider instances considered below the body is part of the input description, so the same terminology also covers degenerate short cases.

For completeness, a tree decomposition of $G$ is a tree whose nodes carry vertex sets, called bags, such that every graph vertex occurs in a connected set of bags and every graph edge has both endpoints in some bag. Its width is one less than the maximum bag size, and the treewidth $\operatorname{tw}(G)$ is the minimum width. The tree-depth $\operatorname{td}(G)$ is the minimum height of a rooted forest whose closure contains $G$, where the closure joins each vertex to all of its ancestors. We use these parameters only to state structural consequences of our reductions.

Let $C=\{1,\ldots,c\}$ be the color set. A Matching-Match instance $\mathcal I$ consists of a host graph $G=(V,E)$, exactly $m$ sticks, and a partial coloring $\mathcal C_0:V\to\{0,1,\ldots,c\}$, where $\mathcal C_0(v)=0$ means that $v$ is initially uncolored. Each stick is identified with its unordered type $\{i,j\}$. For $1\le i\le j\le c$, let $s_{ij}$ be the number of sticks of type $\{i,j\}$.

A final coloring $\phi:V\to C$ \emph{extends} $\mathcal C_0$ if $\phi(v)=\mathcal C_0(v)$ whenever $\mathcal C_0(v)\ne0$. For such a coloring, let $e_{ij}(\phi)$ be the number of graph edges whose endpoint colors form the unordered pair $\{i,j\}$. Since the number of sticks equals the number of graph edges, the placement of individual sticks is irrelevant once these multiplicities agree.

\begin{observation}\label{obs:hist}
An instance is feasible if and only if there exists a final coloring $\phi$ extending $\mathcal C_0$ such that $e_{ij}(\phi)=s_{ij}$ for every $1\le i\le j\le c$.
\end{observation}

An instance is \emph{completely uncolored} if $\mathcal C_0(v)=0$ for every $v\in V$. The counting problem \CountFMMP{} asks for the number of feasible final vertex colorings; permutations of physically identical sticks are not counted separately.

\section{Dense and Multipartite Graphs}\label{sec:dense}

\subsection{Complete $k$-Partite Graphs for Fixed $k$}\label{sec:multipartite-fixed-k}

We first show that fixing the number of parts yields polynomial-time solvability even when both the number of colors and the precoloring are unrestricted.

\begin{theorem}\label{thm:fixed-k}
For every fixed $k$, \FMMP{} on complete $k$-partite host graphs $K_{n_1,\ldots,n_k}$ can be solved in time $O\!\left(c^2k+(n+1)^{k^2}(c^2k+ck^2+k^3)\right)$.
\end{theorem}

\begin{proof}
The case $k=1$ is trivial because the host graph is edgeless, so assume $k\ge2$. Let the parts be $V_1,\ldots,V_k$. For each color $a$, define $x_a=(x_{1a},\ldots,x_{ka})^{\mathsf T}$, where $x_{ra}$ is the number of vertices in $V_r$ that receive color $a$. Let $J_k$ denote the $k\times k$ all-ones matrix and $I_k$ the $k\times k$ identity matrix, and put $B=J_k-I_k$.

For distinct colors $a,b$, the number of host edges of type $\{a,b\}$ is $x_a^{\mathsf T}Bx_b$: every product $x_{ra}x_{sb}$ with $r\ne s$ counts exactly the edges having color $a$ in part $r$ and color $b$ in part $s$. For one color $a$, the expression $x_a^{\mathsf T}Bx_a$ counts each monochromatic edge twice, so the number of $\{a,a\}$ edges is $\frac12x_a^{\mathsf T}Bx_a$. Define the symmetric matrix $Q\in\mathbb Z^{c\times c}$ by $Q_{ab}=s_{ab}$ for $a\ne b$ and $Q_{aa}=2s_{aa}$. If $X=[x_1\ \cdots\ x_c]$, Observation~\ref{obs:hist} shows that feasibility is equivalent to $Q=X^{\mathsf T}BX$, together with $x_{ra}\ge0$, the part-size equations $\sum_a x_{ra}=n_r$, and the lower bounds $x_{ra}\ge p_{ra}$, where $p_{ra}$ is the number of vertices in $V_r$ precolored $a$.

Let $W=\operatorname{span}\{x_1,\ldots,x_c\}$ and let $\mathbf n=(n_1,\ldots,n_k)^{\mathsf T}=\sum_a x_a$. We claim that the bilinear form induced by $B$ is nondegenerate on $W$. Suppose that $0\ne z\in W$ lies in its radical. Since $\mathbf n\in W$, both $z^{\mathsf T}Bz=0$ and $z^{\mathsf T}B\mathbf n=0$. Write $\sigma=\sum_i z_i$. Because $B=J_k-I_k$, we have $0=z^{\mathsf T}Bz=\sigma^2-\sum_i z_i^2$, hence $|\sigma|=\|z\|_2$ and in particular $\sigma\ne0$. If $\sigma>0$, then $|z_i|\le\sigma$ for every $i$, so every coordinate of $Bz$ equals $\sigma-z_i\ge0$. Moreover $Bz\ne0$: otherwise every $z_i=\sigma$, implying $\sigma=k\sigma$, impossible for $k\ge2$ and $\sigma\ne0$. Since every coordinate of $\mathbf n$ is positive, $(Bz)^{\mathsf T}\mathbf n>0$, contradicting $z^{\mathsf T}B\mathbf n=0$. The case $\sigma<0$ is symmetric after replacing $z$ by $-z$.

Thus the restriction of the form to $W$ is nondegenerate, and therefore $r:=\rank Q=\dim W\le k$. A symmetric matrix over a field of characteristic different from two has a nonsingular principal minor of order equal to its rank, so in polynomial time we can choose an index set $I=\{i_1,\ldots,i_r\}$ such that $Q[I,I]$ is nonsingular.

We enumerate the $r$ vectors $x_{i_1},\ldots,x_{i_r}$. Each has $k$ coordinates, with coordinate $p$ in $\{0,\ldots,n_p\}$, so the number of assignments is at most $\prod_{j=1}^r\prod_{p=1}^k(n_p+1)\le(n+1)^{kr}\le(n+1)^{k^2}$. We discard an assignment if a basis vector violates a precolor lower bound or if $x_{i_p}^{\mathsf T}Bx_{i_q}\ne Q_{i_pi_q}$ for some $p,q$.

For a surviving assignment, let $Y=[x_{i_1}\cdots x_{i_r}]$. In any feasible completion the columns of $Y$ form a basis of $W$. Hence every remaining color vector has the form $x_j=Y\lambda_j$. Taking its inner products with the basis gives $Q[I,j]=Y^{\mathsf T}Bx_j=Q[I,I]\lambda_j$, and therefore $\lambda_j=Q[I,I]^{-1}Q[I,j]$. Thus every $x_j$ is uniquely determined. We reconstruct all of them by exact rational arithmetic and check that every coordinate is integral and nonnegative, that all precolor lower bounds hold, that $\sum_jx_j=\mathbf n$, and finally that $Q=X^{\mathsf T}BX$.

If these checks succeed, then in each part $V_r$ we color exactly $x_{ra}$ vertices with color $a$. The lower bounds ensure that the precolored vertices can be included, while the remaining vertices of a part are interchangeable. The identity $Q=X^{\mathsf T}BX$ then gives exactly the required edge-type multiplicities, so Observation~\ref{obs:hist} yields a feasible solution. Conversely, every feasible solution supplies one of the enumerated basis assignments and passes all checks.

Using symmetric Gaussian elimination and stopping as soon as $k+1$ independent pivots are found, we either reject because $\rank Q>k$ or obtain $\rank Q$ and a nonsingular principal submatrix in $O(c^2k)$ time. There are at most $(n+1)^{k^2}$ basis assignments. For each assignment, checking the basis inner products takes $O(k^3)$ time, reconstructing all $c$ color-count vectors takes $O(ck^2)$ time, and verifying $Q=X^{\mathsf T}BX$ takes $O(c^2k)$ time. Hence the total running time is $O\!\left(c^2k+(n+1)^{k^2}(c^2k+ck^2+k^3)\right)$. \qed
\end{proof}

\begin{corollary}\label{cor:multipartite-xp}
\FMMP{} on complete multipartite graphs belongs to \XP{} when parameterized by the number $k$ of parts.
\end{corollary}

\subsection{A Sharp Complement-Degree Threshold}

We next measure how far the host graph is from being complete. The positive side consists exactly of complete graphs with a matching of edges deleted.

\begin{theorem}\label{thm:comp-one}
If $\Delta(\overline G)\le1$, then \FMMP{} is solvable for arbitrary numbers of colors and arbitrary precoloring in time $O(m+c^2+(n+c)^{1+o(1)})$.
\end{theorem}

\begin{proof}
Write $G=K_n-M$, where $M$ is a matching of size $q$. Thus $u=n-2q$ vertices have degree $n-1$ and the $2q$ endpoints of $M$ have degree $n-2$. For each color $a$, let $D_a=2s_{aa}+\sum_{b\ne a}s_{ab}$ be the number of stick endpoints of color $a$. Let $x_a$ be the number of degree-$(n-1)$ vertices colored $a$, let $y_a$ be the number of degree-$(n-2)$ vertices colored $a$, and put $t_a=x_a+y_a$. Every feasible coloring satisfies $D_a=(n-1)x_a+(n-2)y_a=(n-2)t_a+x_a$.

Assume first that $q\ge2$. Then $u=n-2q<n-2$, so $0\le x_a\le u<n-2$. Consequently $x_a$ is the unique integer in $\{0,\ldots,u\}$ satisfying $x_a\equiv D_a\pmod{n-2}$. Once $x_a$ is known, $t_a=(D_a-x_a)/(n-2)$ and $y_a=t_a-x_a$ are forced. We reject if these values are not nonnegative integers or if their sums are inconsistent with $u$ and $2q$. If $q=0$, the host is complete and the color multiplicities are determined directly from the endpoint counts. If $q=1$, we enumerate the $O(c^2)$ possible color pairs on the unique missing edge and then determine all remaining color counts; this case is also polynomial.

For $a\ne b$, a complete graph on a vertex multiset containing $t_a$ vertices of color $a$ and $t_b$ vertices of color $b$ would contain $t_at_b$ pairs of type $\{a,b\}$, while it would contain $\binom{t_a}{2}$ pairs of type $\{a,a\}$. Since the only nonedges of $G$ are the edges of $M$, the missing matching must contain exactly $r_{ab}=t_at_b-s_{ab}$ nonedges of type $\{a,b\}$ for $a\ne b$, and $r_{aa}=\binom{t_a}{2}-s_{aa}$ nonedges of type $\{a,a\}$. We reject if any $r_{ab}<0$ or if $\sum_{a\le b}r_{ab}\ne q$.

It remains to realize the missing-edge types while respecting endpoint precolors. We first handle every missing edge whose two endpoints are precolored. If $f_{ab}$ such edges have color pair $\{a,b\}$, then necessarily $f_{ab}\le r_{ab}$; otherwise we reject. Subtract these forced copies from $r_{ab}$. For each color $a$, let $h_a$ be the number of remaining missing edges with exactly one precolored endpoint, whose fixed color is $a$.

The remaining choice is compressed to a flow on colors and pair types. Create a source, one node for each type $\{a,b\}$ with positive remaining multiplicity $r_{ab}$, one node for each color, and a sink. Add an arc from the source to the type node $\{a,b\}$ of capacity $r_{ab}$, and arcs from this type node to color nodes $a$ and $b$, each of capacity $r_{ab}$; for a loop type $\{a,a\}$ there is only one such arc. Finally, add an arc from color node $a$ to the sink of capacity $h_a$. Let $H=\sum_a h_a$. A flow of value $H$ selects, for every half-precolored missing edge, one unused pair type containing its fixed color. The selected types can then be assigned arbitrarily to the corresponding half-precolored edges of each color, and every unselected type can be placed on an edge whose two endpoints are uncolored. Conversely, every feasible coloring induces such a flow. Thus this flow is feasible exactly when the missing matching can be colored consistently.

The full-degree vertices are independent of this flow. If $p_a$ of them are precolored $a$, they can realize the forced counts $x_a$ exactly when $p_a\le x_a$ for every $a$ and $\sum_a(x_a-p_a)$ equals the number of uncolored full-degree vertices.

After the $O(c^2)$ type multiplicities $r_{ab}$ are computed, at most $q\le n/2$ of them are positive. Hence the flow network has $O(n+c)$ arcs and vertices, with integral capacities bounded by $n$. The deterministic exact max-flow algorithm of van den Brand et al.~\cite{BrandEtAl2023Flow} runs in $M^{1+o(1)}$ time on a network with $M$ arcs and polynomially bounded integral capacities, so the flow step takes $(n+c)^{1+o(1)}$ time. Computing the input histograms and all remaining checks takes $O(m+c^2)$ time, giving the stated bound. \qed
\end{proof}

The next result shows that increasing the complement degree by one already makes the problem hard.

\begin{theorem}\label{thm:comp-two}
\FMMP{} is NP-complete on completely uncolored graphs whose complement has maximum degree $2$. Hardness already holds when the complement is a disjoint union of triangles.
\end{theorem}

\begin{proof}
Membership in NP follows by guessing the final coloring and checking its edge-type histogram. We reduce from \textsc{Triangle Decomposition}, which asks whether the edge set of a given simple graph can be partitioned into edge-disjoint triangles and is NP-complete~\cite{Holyer1981}. Let $H=(V,E)$ be the source graph and let $m=|E|$. A triangle decomposition requires $3\mid m$ and every $d_H(v)$ to be even; if either condition fails, output a fixed NO-instance. Put $t_v=d_H(v)/2$. Since $\sum_vd_H(v)=2m$, we have $\sum_vt_v=m$.

Construct the host graph $G=K_{3,3,\ldots,3}$ with $m/3$ parts. Thus $G$ has $m$ vertices, every vertex has degree $m-3$, and the nonedges of $G$ are exactly the three pairs inside each part. Introduce one color, also denoted $v$, for every source vertex $v\in V(H)$. For distinct source vertices $u,v$, create $s_{uv}=t_ut_v-\mathbf 1_{uv\in E(H)}$ sticks of type $\{u,v\}$, and create $s_{vv}=\binom{t_v}{2}$ sticks of type $\{v,v\}$. These numbers are nonnegative. Their total is $\binom{m}{2}-m=|E(G)|$, because they are obtained from all pairs in a multiset with $t_v$ copies of each color $v$ by deleting one pair for every source edge.

The number $D_v$ of stick endpoints of color $v$ is
$D_v=2\binom{t_v}{2}+\sum_{u\ne v}(t_ut_v-\mathbf1_{uv\in E(H)})=t_v(m-3)$.
For $m>3$, every host vertex has degree $m-3$, so any feasible coloring must use color $v$ on exactly $D_v/(m-3)=t_v$ vertices. The constant-size case $m=3$ can be decided directly.

Now compare the complete graph on these colored vertices with the actual host graph. For $u\ne v$ there are $t_ut_v$ vertex pairs of colors $u,v$, but only $s_{uv}$ host edges of this type; hence exactly one $uv$ pair is a nonedge precisely when $uv\in E(H)$. For one color $v$, all $\binom{t_v}{2}$ same-color pairs are host edges, so no nonedge has equal endpoint colors. Since every host part consists of three mutually nonadjacent vertices, its three colors are distinct, say $a,b,c$, and the three missing pairs $ab,bc,ca$ are all edges of $H$. Thus each part determines a source triangle. Every source edge appears as a missing pair exactly once, so these triangles form a triangle decomposition of $H$.

Conversely, given a triangle decomposition of $H$, color one host part by the three vertices of each source triangle. Each source vertex $v$ appears in exactly $d_H(v)/2=t_v$ triangles, so the required color multiplicities are correct. A pair of differently colored host vertices is missing exactly when the corresponding source edge belongs to the triangle assigned to that part; therefore the host-edge histogram is precisely the prescribed stick histogram. \qed
\end{proof}

\begin{corollary}[Exact complement-degree threshold]\label{cor:comp}
For unrestricted numbers of colors, \FMMP{} is polynomial-time solvable when $\Delta(\overline G)\le1$, whereas it is NP-complete already when $\Delta(\overline G)=2$.
\end{corollary}

\subsection{Uniform Complete Multipartite Graphs}

A complete multipartite graph is \emph{uniform} if all parts have the same size. The complement-degree result settles part sizes one and two; larger fixed part sizes are hard.

\begin{theorem}\label{thm:uniform}
Fix $r\ge1$ and let the host graph be $K_{r,r,\ldots,r}$ with an arbitrary number of parts. Then \FMMP{} is polynomial-time solvable for $r\le2$ and NP-complete for every fixed $r\ge3$. The hardness holds without precolored vertices.
\end{theorem}

\begin{proof}
For $r=1$ the host graph is complete and the result follows from the polynomial algorithm of Iburi and Uehara~\cite{IburiUehara2024}. For $r=2$, the complement is a perfect matching, so Theorem~\ref{thm:comp-one} applies.

Fix $r\ge3$. We reduce from \textsc{$K_r$-Decomposition}, which asks whether the edge set of a given graph can be partitioned into copies of $K_r$ and is NP-complete for every fixed $r\ge3$~\cite{Holyer1981}. Let $H=(V,E)$ be the source graph and $m=|E|$. In every $K_r$-decomposition, each vertex degree is divisible by $r-1$ and $m$ is divisible by $\binom r2$; if these necessary conditions fail, output a fixed NO-instance. Put $t_v=d_H(v)/(r-1)$ and $N=\sum_vt_v=2m/(r-1)$. The host graph consists of $N/r=m/\binom r2$ parts of size $r$, so it has $N$ vertices and every host vertex has degree $N-r$.

Introduce one color for every source vertex. For $u\ne v$, create $s_{uv}=t_ut_v-\mathbf1_{uv\in E(H)}$ sticks, and create $s_{vv}=\binom{t_v}{2}$ sticks. As in the triangle reduction, these are all pair multiplicities of a multiset with $t_v$ copies of color $v$, except that one pair is removed for every source edge. Therefore the number of sticks equals $|E(K_{r,r,\ldots,r})|$.

For every color $v$, its number of stick endpoints is $D_v=2\binom{t_v}{2}+\sum_{u\ne v}(t_ut_v-\mathbf1_{uv\in E(H)})=t_v(N-1)-d_H(v)=t_v(N-r)$. Since the host is $(N-r)$-regular, every feasible coloring uses color $v$ on exactly $t_v$ host vertices.

Consequently, for each source edge $uv$ exactly one pair of host vertices colored $u,v$ must be a nonedge, while no nonedge may have equal endpoint colors. The nonedges of the host are precisely the pairs lying in the same size-$r$ part. Hence each part receives $r$ distinct source colors, and every pair among these colors is an edge of $H$. Each part therefore identifies a copy of $K_r$ in $H$. Since every source edge occurs as a missing pair exactly once, the parts yield a $K_r$-decomposition.

Conversely, from a $K_r$-decomposition of $H$, color one host part by the $r$ source vertices of each decomposition clique. Vertex $v$ appears in exactly $d_H(v)/(r-1)=t_v$ cliques, and the same pair-count argument shows that the resulting host-edge histogram equals the prescribed stick histogram. \qed
\end{proof}

\subsection{Hardness at Diameter Two}

The preceding hardness examples are highly disconnected in the complement. We now show that even connected host graphs of the smallest nontrivial diameter can be hard.

\begin{theorem}\label{thm:diam2}
\FMMP{} is NP-complete on completely uncolored cographs that are connected and have diameter $2$. Moreover, the hard host graphs are block graphs with a universal vertex, treewidth $3$, and clique number $4$.
\end{theorem}

\begin{proof}
Membership in NP is immediate. We reduce from \textsc{Triangle Decomposition}, which asks whether the edge set of a given simple graph can be partitioned into triangles and is NP-complete~\cite{Holyer1981}. Let $H=(V,E)$ be an instance with $m=|E|$. If $3\nmid m$ or some source vertex has odd degree, output a fixed NO-instance. Put $t=m/3$ and construct the host graph $G=K_1\vee(tK_3\dot\cup mK_1)$. Let $r$ be the universal vertex. The graph has $m$ internal triangle edges, $m$ edges from $r$ to triangle vertices, and $m$ edges from $r$ to leaves, hence $3m$ edges in total. It is a cograph and a block graph, is connected with diameter two, and has clique number four and treewidth three.

Introduce a special color $S$, one color $c_v$ for every source vertex $v$, and fresh colors $z_1,\ldots,z_m$. Create one stick $\{c_u,c_v\}$ for every source edge $uv$, create $d_H(v)/2$ sticks of type $\{S,c_v\}$ for every $v$, and create one stick $\{S,z_i\}$ for every $i$. The first family has $m$ sticks, the second has $\sum_vd_H(v)/2=m$ sticks, and the third has $m$, so the total is $3m=|E(G)|$.

Color $S$ occurs on exactly $2m$ stick endpoints. A source color $c_v$ occurs on $d_H(v)+d_H(v)/2\le3m/2<2m$ endpoints, and each $z_i$ occurs once. Since the universal vertex $r$ has degree $2m$, its final color must occur on at least $2m$ stick endpoints; hence $r$ is forced to color $S$. There are exactly $2m$ sticks containing $S$, so every one of them is used on an edge incident with $r$.

After these sticks are used, only the $m$ source-edge sticks remain for the $m$ internal edges of the host triangles. Therefore no triangle vertex can receive a color $z_i$, because its two internal incident edges would then require sticks containing $z_i$, and none remains. Thus every triangle vertex receives some source color $c_v$. It follows that the $m$ universal-to-triangle edges consume all $m$ sticks of type $\{S,c_v\}$, and the $m$ universal-to-leaf edges consume the $m$ sticks $\{S,z_i\}$.

Inside a host triangle no two vertices can share a source color, because no source-edge stick has type $\{c_v,c_v\}$. If the three colors are $c_a,c_b,c_c$, its internal edges require precisely the three sticks $\{c_a,c_b\}$, $\{c_b,c_c\}$, and $\{c_c,c_a\}$, so $ab,bc,ca\in E(H)$. Hence every host triangle determines a source triangle. Every source-edge stick is used exactly once, so these source triangles partition $E(H)$.

Conversely, given a triangle decomposition of $H$, color one host triangle by each source triangle. A source vertex $v$ occurs in exactly $d_H(v)/2$ decomposition triangles, so the $\{S,c_v\}$ sticks fit exactly on the corresponding edges from $r$. Put the $\{S,z_i\}$ sticks on the leaf edges. This uses every stick once and gives a feasible coloring. \qed
\end{proof}

\begin{corollary}[Exact connected-diameter threshold]
Connected graphs of diameter $1$ are polynomial-time solvable for \FMMP{}, whereas NP-completeness already holds at diameter $2$.
\end{corollary}

\begin{proof}
A connected graph has diameter one exactly when it is complete, and complete host graphs are polynomial-time solvable~\cite{IburiUehara2024}. Theorem~\ref{thm:diam2} supplies hardness at diameter two. \qed
\end{proof}

\section{Parameterized Complexity on Cographs}\label{sec:cographs}

\subsection{W[1]-Hardness Parameterized by the Number of Colors}

Dumitru, Micl\u au\c s, and Popa give an $n^{O(c^2)}$ exact algorithm for fixed $c$ on cographs~\cite{DumitruMiclausPopaMax}. We show that the dependence on $c$ cannot in general be replaced by an FPT dependence unless FPT equals W[1].

\begin{theorem}\label{thm:cograph-w1}
\FMMP{} is W[1]-hard on cographs parameterized by the number of colors $c$.
\end{theorem}

\begin{proof}
We reduce from \textsc{Equitable Coloring} on cographs, which asks whether a graph $H$ has a proper $k$-coloring in which every color class has size either $\lfloor |V(H)|/k\rfloor$ or $\lceil |V(H)|/k\rceil$; parameterized by $k$, this problem is W[1]-hard even on cographs~\cite{GomesGuedesSantos2023}. Let $(H,k)$ be such an instance with $H=(V,E)$, and write $n=|V|$ and $m=|E|$. We first make $n$ divisible by $k$. If $n=qk+r$ with $0<r<k$, add a disjoint clique $K_{k-r}$. The resulting graph is still a cograph. It has an equitable $k$-coloring if and only if $H$ does: an equitable coloring of $H$ uses $r$ colors $q+1$ times and the other $k-r$ colors $q$ times, and after permuting color names the new clique can use exactly the latter $k-r$ colors. We may therefore assume $n=qk$.

Introduce one additional color $\alpha$ and construct the host graph $G'=(K_1\vee H)\dot\cup B K_2$, where $B=m(\binom{k}{2}-1)$. For the nontrivial W[1]-hard instances we may assume $k\ge2$, so $B\ge0$. Let $p$ be the universal vertex in $K_1\vee H$ and precolor $p$ with $\alpha$. Cographs are closed under join and disjoint union, so $G'$ is a cograph.

For every ordinary color $i\in[k]$, create exactly $q$ sticks of type $\{\alpha,i\}$. For every unordered pair $1\le i<j\le k$, create exactly $m$ sticks of type $\{i,j\}$. There are no monochromatic sticks. The number of sticks is $qk+m\binom{k}{2}=n+m\binom{k}{2}$. The host has $n$ edges incident with $p$, the $m$ edges of $H$, and $B$ isolated buffer edges, so $|E(G')|=n+m+B=n+m\binom{k}{2}$ as required.

Suppose first that $H$ has an equitable proper $k$-coloring. Because $n=qk$, every color is used exactly $q$ times. Color the copy of $H$ accordingly. Put the $q$ sticks $\{\alpha,i\}$ on the edges from $p$ to the $q$ vertices of color $i$. Every source edge of $H$ is bichromatic, so it consumes one of the $m$ available sticks of its color pair. For each pair $i<j$, the unused copies of $\{i,j\}$ are placed arbitrarily on buffer edges. Their total number is exactly $B$, so all sticks and edges are used.

Conversely, consider a feasible Matching-Match coloring. The precolored vertex $p$ has degree $n$, and the construction contains exactly $n$ sticks containing $\alpha$. Therefore all such sticks are used on edges incident with $p$. Since there are exactly $q$ sticks $\{\alpha,i\}$, precisely $q$ vertices of the copy of $H$ receive each color $i$. Every stick not containing $\alpha$ is bichromatic, so every edge of $H$ has differently colored endpoints. Hence the induced coloring of $H$ is proper and equitable.

The reduction maps parameter $k$ to exactly $k+1$ Matching-Match colors and is computable in polynomial time. Therefore it is a parameterized reduction. \qed
\end{proof}

\section{Uncolored Paths and Cycles}\label{sec:paths-cycles}

\subsection{Eulerian Characterization}

For a stick multiset $S$, define its \emph{color multigraph} $H_S$: the vertices are the colors that occur on sticks, and every stick of type $\{a,b\}$ becomes one undirected edge between colors $a$ and $b$, with parallel edges and loops allowed. An \emph{Euler trail} is a walk that uses every edge exactly once; an \emph{Euler circuit} is an Euler trail whose initial and final vertices coincide.

\begin{theorem}\label{thm:euler}
If no host vertex is precolored, \FMMP{} is solvable in $O(m+c)$ time on paths and cycles. A path instance is feasible if and only if $H_S$ has an Euler trail, and a cycle instance is feasible if and only if $H_S$ has an Euler circuit.
\end{theorem}

\begin{proof}
Let the host path be $v_0v_1\cdots v_m$. A feasible solution gives a color sequence $c_0,c_1,\ldots,c_m$, where the stick on edge $v_{i-1}v_i$ has type $\{c_{i-1},c_i\}$. Reading the sticks in path order therefore traverses every edge of $H_S$ exactly once, so it is an Euler trail. Conversely, an Euler trail of $H_S$ gives an ordering and orientation of all sticks in which consecutive sticks share the appropriate color. Assigning the successive color vertices of the trail to $v_0,\ldots,v_m$ produces a feasible host coloring. The same argument for a cycle gives an Euler circuit.

An undirected multigraph has an Euler circuit exactly when all non-isolated vertices lie in one connected component and every degree is even; it has an Euler trail exactly when the same connectivity condition holds and the number of odd-degree vertices is zero or two. Degrees and connectivity of $H_S$ are computed by one scan of the sticks, giving $O(m+c)$ time. \qed
\end{proof}

\begin{corollary}\label{cor:endpoints}
A path remains solvable in $O(m+c)$ time when only its two endpoints may be precolored: one asks whether $H_S$ has an Euler trail whose initial and terminal colors are compatible with the prescribed endpoint colors.
\end{corollary}

\subsection{Counting Feasible Colorings on Paths and Cycles}

The decision problem on these host graphs is linear-time solvable, but the corresponding counting problem is hard.

\begin{theorem}\label{thm:count-cycle}
\CountFMMP{} is $\#P$-complete even on completely uncolored cycles.
\end{theorem}

\begin{proof}
Membership in $\#P$ follows because a proposed final coloring can be checked in polynomial time using Observation~\ref{obs:hist}. For hardness, we reduce from counting Euler circuits in a connected undirected Eulerian graph, which is $\#P$-complete by Brightwell and Winkler~\cite{BrightwellWinkler2005}. Rooting an Euler circuit at a specified position and choosing a traversal direction changes the count only by a polynomial-time computable factor, so we use the equivalent rooted, oriented version. Let $H=(W,F)$ be a connected simple Eulerian graph with $m=|F|$. Give every source vertex $w$ its own color $c_w$. Use the labeled cycle $C_m=v_0v_1\cdots v_{m-1}v_0$ as the host, with no precolored vertices, and create one stick $\{c_x,c_y\}$ for every source edge $xy\in F$.

Because $H$ is simple and all source vertices have distinct colors, every stick type identifies one source edge. A feasible host coloring $\phi(v_i)=c_{w_i}$ therefore gives the edge sequence $w_0w_1,w_1w_2,\ldots,w_{m-1}w_0$, which uses every edge of $H$ exactly once: it is a rooted, oriented Euler circuit. Conversely, every such Euler circuit gives exactly one coloring of the labeled cycle. Thus the construction is parsimonious from the rooted, oriented version of \#\textsc{Eulerian Circuit}, and the result follows. \qed
\end{proof}

\begin{theorem}\label{thm:count-path}
\CountFMMP{} is $\#P$-complete even on completely uncolored paths.
\end{theorem}

\begin{proof}
Membership in $\#P$ is as above. For hardness, we reduce from the rooted, oriented Euler-circuit counting problem: given a connected undirected Eulerian graph $H$ and a specified root vertex $r$, count Euler circuits that start at $r$ with a chosen traversal direction. This variant is $\#P$-complete by the result of Brightwell and Winkler, up to the polynomial-time computable rooting/orientation factor~\cite{BrightwellWinkler2005}. Let $H$ be a connected simple Eulerian graph with specified root $r$. Add two new vertices $s,t$ and the two pendant edges $sr$ and $rt$, obtaining a graph $H'$. The only odd-degree vertices of $H'$ are $s$ and $t$, so every Euler trail starts at one of them and ends at the other.

Give every vertex of $H'$ a distinct color, create one stick for every edge of $H'$, and take a labeled path with $|E(H')|$ edges as the completely uncolored host graph. As in Theorem~\ref{thm:count-cycle}, feasible colorings are in bijection with oriented Euler trails of $H'$. Every trail from $s$ to $t$ begins with $sr$ and ends with $rt$; deleting those two edges leaves a rooted, oriented Euler circuit of $H$ based at $r$. Conversely, every rooted, oriented Euler circuit of $H$ extends uniquely to an $s$--$t$ Euler trail of $H'$. Reversing the host path gives the corresponding $t$--$s$ trail, so the number of feasible path colorings is exactly twice the number of rooted, oriented Euler circuits of $H$. Division by two recovers the source count. \qed
\end{proof}

\subsection{Polynomial Counting on Stars and Complete Graphs}

The counting separation is not universal among decision-easy graph classes.

\begin{theorem}\label{thm:count-star-complete}
Counting feasible final colorings is polynomial-time solvable on a star and on a complete graph, even with arbitrary precoloring.
\end{theorem}

\begin{proof}
First let the host be a star with center $z$ and $m$ leaves. Fix a candidate center color $a$ compatible with the precolor of $z$. Every stick must then contain $a$; otherwise this candidate is impossible. If $b\ne a$, let $q_b=s_{\min\{a,b\},\max\{a,b\}}$, and put $q_a=s_{aa}$. These values force exactly $q_b$ leaves to receive color $b$. Let $p_b$ be the number of leaves precolored $b$. The candidate $a$ is feasible exactly when every $q_b\ge p_b$ and $\sum_bq_b=m$, and then the number of extensions is the multinomial coefficient $(m-\sum_bp_b)!/\prod_b(q_b-p_b)!$. We sum this value over all center colors allowed by the center precolor. If $m>0$, any feasible center color belongs to the endpoint-color set of every stick, so there are at most two candidates after intersecting these sets. Hence one scan of the sticks and leaves suffices, followed by the required factorial arithmetic.

Now let the host be $K_n$ with $n\ge2$. Put $D_i=2s_{ii}+\sum_{j\ne i}s_{ij}$. Every vertex has degree $n-1$, so if $x_i$ vertices receive color $i$, necessarily $D_i=(n-1)x_i$ and hence $x_i=D_i/(n-1)$. We reject unless all $x_i$ are nonnegative integers summing to $n$ and the required edge multiplicities satisfy $s_{ii}=\binom{x_i}{2}$ and $s_{ij}=x_ix_j$ for every $i<j$. Let $p_i$ be the number of vertices precolored $i$. If every $p_i\le x_i$, the uncolored vertices can be assigned the remaining multiplicities in exactly $(n-\sum_ip_i)!/\prod_i(x_i-p_i)!$ ways; otherwise the count is zero. Computing all $D_i$ and checking all pair multiplicities takes $O(n+c^2)$ time after the stick multiplicities have been read. The required factorials and multinomial coefficients have polynomial encoding length and can be evaluated exactly in polynomial time. The cases $n\le1$ are immediate. \qed
\end{proof}

\section{Disconnected Uncolored Host Graphs}\label{sec:disconnected}

\subsection{A General Decomposition Transfer}

The Eulerian characterization above exploits the fact that one connected path or cycle consumes the stick pool in a single coherent walk. For disconnected host graphs, the common stick pool must instead be partitioned among components. This connects Matching-Match directly to graph-decomposition problems.

\begin{theorem}\label{thm:h-transfer}
Let $H$ be a fixed connected simple graph with $|E(H)|\ge3$ and $\operatorname{diam}(H)\le2$. Then \FMMP{} is NP-complete on completely uncolored host graphs that are disjoint unions of copies of $H$.
\end{theorem}

\begin{proof}
Membership in NP is immediate. We reduce from \textsc{$H$-Decomposition}: for the fixed graph $H$, the input is a graph $F$ and the question is whether $E(F)$ can be partitioned into edge-disjoint subgraphs isomorphic to $H$. Dor and Tarsi proved this problem NP-complete for every fixed connected $H$ with at least three edges~\cite{DorTarsi1997}. Let $F=(V,E)$ be an instance. If $|E|$ is not divisible by $|E(H)|$, output a fixed NO-instance; otherwise put $t=|E|/|E(H)|$. Introduce one color $c_v$ for every source vertex $v\in V$ and one stick $\{c_u,c_v\}$ for every source edge $uv\in E$. Let the host graph be the disjoint union of $t$ copies of $H$, with no precolored vertices.

If $F$ has an $H$-decomposition, color each host copy according to the corresponding source copy of $H$ and place the source-edge sticks on the matching host edges. This is feasible.

Conversely, consider one connected host component and suppose it has a feasible coloring. Two adjacent host vertices cannot receive the same source color because the source graph is simple and hence there is no loop stick $\{c_v,c_v\}$. Suppose two nonadjacent host vertices $x,y$ received the same color $c_v$. Since $\operatorname{diam}(H)\le2$, they have a common neighbor $z$ in this component. If $z$ receives color $c_u$, then both host edges $xz$ and $zy$ require the same stick type $\{c_v,c_u\}$. Because the source graph is simple, only one such stick exists, a contradiction. Therefore the coloring is injective on every host component.

The $|E(H)|$ sticks used by one component consequently correspond to a source subgraph isomorphic to $H$. Different host components use disjoint sticks, and every stick is used once, so these source copies partition $E(F)$. Thus they form an $H$-decomposition. \qed
\end{proof}

\begin{corollary}\label{cor:td2}
Completely uncolored \FMMP{} is NP-complete on graphs of tree-depth $2$.
\end{corollary}

\begin{proof}
Apply Theorem~\ref{thm:h-transfer} with $H=K_{1,3}$. Every connected component is a nontrivial star, so the resulting host graph has tree-depth exactly two. \qed
\end{proof}

\subsection{Star Forests with Precolored Centers}

The previous corollary shows hardness on completely uncolored star forests. Precoloring every center removes the coupling inside each component.

\begin{theorem}\label{thm:star-forest}
Let the host graph be a star forest with $m$ edges, and suppose every center of a nontrivial star is precolored. Leaves may be precolored or uncolored. Then \FMMP{} is solvable in $(m+c)^{1+o(1)}$ time.
\end{theorem}

\begin{proof}
For every host edge whose leaf is precolored, both endpoint colors are fixed. Thus, if its center has color $a$ and its leaf has color $b$, the edge must consume one stick of type $\{a,b\}$. Subtract all such forced copies, rejecting if any required multiplicity is unavailable. For each color $a$, let $h_a$ be the number of remaining edges whose center has color $a$; all their leaves are uncolored.

Every remaining stick must be placed on one of these edges, and a stick of type $\{a,b\}$ can be used on an edge centered at $a$ or at $b$. We solve this assignment without constructing the quadratic compatibility graph. Create a source, one node for every remaining stick type $\{a,b\}$ of positive multiplicity $q_{ab}$, one node for every color, and a sink. Add an arc from the source to the type node with capacity $q_{ab}$, arcs from the type node to color nodes $a$ and $b$ with capacity $q_{ab}$ (only one arc for a loop), and an arc from color node $a$ to the sink with capacity $h_a$. Since $\sum_a h_a$ equals the number of remaining sticks, a flow of this value assigns every stick copy to exactly one compatible center color and fills every quota $h_a$.

Such a flow gives a feasible puzzle solution by assigning the selected stick copies arbitrarily to the uncolored leaves adjacent to centers of the corresponding color. Conversely, every feasible solution induces exactly this flow. There are at most $m$ positive stick types, so the network has $O(m+c)$ arcs and polynomially bounded integral capacities. By the deterministic exact max-flow algorithm of van den Brand et al.~\cite{BrandEtAl2023Flow}, feasibility is decided in $(m+c)^{1+o(1)}$ time. \qed
\end{proof}

\subsection{A Sharp Threshold for Tiny Components}

We now classify completely uncolored host graphs whose connected components contain very few edges.

\begin{theorem}\label{thm:twoedge}
If the host graph is completely uncolored and every connected component contains at most two edges, then \FMMP{} is solvable in $O(m+c)$ time.
\end{theorem}

\begin{proof}
Every nontrivial host component is either one edge $K_2$ or a two-edge path $P_3$. Let $r$ be the number of $P_3$ components. Recall the color multigraph $H_S$, whose edges are the sticks. A $P_3$ requires two sticks sharing the color assigned to its middle vertex, so we must choose $r$ disjoint pairs of incident edges of $H_S$.

Consider an edge-containing connected component $C$ of $H_S$ with $q$ sticks. Let $L(C)$ be the simple compatibility graph whose vertices are these sticks and whose edges join pairs sharing a color. This graph is connected and claw-free: the neighbors of a nonloop stick split according to its two endpoint colors, so among any three neighbors two are adjacent; for a loop stick all neighbors share its unique color. If $q$ is even, $L(C)$ has a perfect matching by Sumner's theorem~\cite{Sumner1974}. If $q$ is odd and $q>1$, delete a non-cut vertex of $L(C)$; the remaining graph is connected, claw-free, and of even order, so it has a perfect matching. Thus $C$ contains exactly $\lfloor q/2\rfloor$ disjoint compatible stick pairs. The case $q=1$ is immediate.

Therefore, if the edge-containing components of $H_S$ have $q_1,\ldots,q_t$ edges, the maximum number of disjoint compatible stick pairs is $\sum_{i=1}^t\lfloor q_i/2\rfloor$. The Matching-Match instance is feasible exactly when this quantity is at least $r$: use one pair for every $P_3$ and put every remaining stick independently on a $K_2$. The components of $H_S$ and their edge counts are found by one graph traversal, in $O(m+c)$ time. \qed
\end{proof}

Allowing one additional edge per component makes the problem hard.

\subsection{Hardness on Uncolored Linear Forests}

\begin{theorem}\label{thm:p4}
\FMMP{} is NP-complete on completely uncolored linear forests in which every connected component is exactly a copy of $P_4$.
\end{theorem}

\begin{proof}
Membership in NP is immediate. We reduce from \textsc{$P_4$-Decomposition}, which asks whether the edge set of a graph can be partitioned into copies of the three-edge path $P_4$ and remains NP-complete on bipartite graphs~\cite{TeypazRapine2008}. Let $H=(V,E)$ be a simple bipartite instance. If $|E|$ is not divisible by three, output a fixed NO-instance; otherwise write $|E|=3t$. Introduce one color $c_v$ for every source vertex $v$ and one stick $\{c_u,c_v\}$ for every source edge $uv$. Let the completely uncolored host graph be the disjoint union of $t$ copies of $P_4$.

A $P_4$-decomposition of $H$ immediately gives a feasible solution by copying the four source-vertex colors of each decomposition path onto one host component.

Conversely, consider one host path, with source colors $c_{v_0},c_{v_1},c_{v_2},c_{v_3}$ in path order. Its three sticks represent the source edge between $v_0$ and $v_1$, the source edge between $v_1$ and $v_2$, and the source edge between $v_2$ and $v_3$. Adjacent equalities are impossible because the source graph has no loops. If $v_0=v_2$, the first two host edges would require two copies of the same source edge, but only one corresponding stick exists; similarly $v_1\ne v_3$. Finally, if $v_0=v_3$, the three source edges would form a triangle, impossible because $H$ is bipartite. Hence all four colors are distinct and the three source edges form a genuine $P_4$. Since every stick is used exactly once, the host components yield a $P_4$-decomposition of $H$. \qed
\end{proof}

\begin{corollary}[Exact maximum-degree threshold]\label{cor:maxdeg}
For completely uncolored, possibly disconnected host graphs with an unrestricted number of colors, maximum degree at most $1$ is polynomial-time solvable, while NP-completeness already holds at maximum degree $2$.
\end{corollary}

\begin{proof}
At maximum degree at most one every nontrivial component is an isolated edge, so every stick can be placed independently. Theorem~\ref{thm:p4} gives NP-completeness at maximum degree two. \qed
\end{proof}

\begin{corollary}[Exact edge-per-component threshold]
For completely uncolored host graphs, components with at most two edges yield a polynomial-time problem, whereas NP-completeness already occurs when every component has exactly three edges.
\end{corollary}

The hard host graphs of Theorem~\ref{thm:p4} are planar and bipartite, have maximum degree two and treewidth one, and every component has four vertices. Together with Theorem~\ref{thm:euler}, this shows that connectivity itself can change the complexity: one arbitrary uncolored path is easy, but a disjoint union of constant-size uncolored paths is NP-complete because the common stick pool creates a global partitioning constraint.

\section{Short Spiders with Unrestrictedly Many Colors}\label{sec:short-spiders}

Dumitru, Micl\u au\c s, and Popa prove that arbitrary precoloring on spiders whose legs have length at most two is fixed-parameter tractable in the number of colors~\cite{DumitruMiclausPopaSpiders}. Here we identify two precoloring patterns that remain polynomial even when the number of colors is part of the input.

\subsection{Precolored Internal Vertices}

\begin{theorem}\label{thm:length2-middle}
Let every spider leg have length at most two, and suppose every internal vertex at distance one from the body on a length-two leg is precolored. Leaves and the body may be precolored or uncolored. If the host spider has $m$ edges, then \FMMP{} is solvable in $c(m+c)^{1+o(1)}$ time.
\end{theorem}

\begin{proof}
We try every possible final color $s$ of the body that is compatible with its precolor; there are at most $c$ choices. Fix one such $s$.

Consider a length-two leg $b-u_i-v_i$, where $u_i$ is precolored $a_i$. The body edge $bu_i$ must consume one stick of type $\{s,a_i\}$. Subtract these forced copies for all length-two legs, rejecting if any multiplicity becomes negative. Every remaining host edge now has one endpoint of fixed color: it is $a_i$ on an outer edge $u_iv_i$ and $s$ on a length-one edge $bv$.

If the leaf of such an edge is precolored $b$, its stick type is forced to be $\{a,b\}$, where $a$ is the fixed color of the other endpoint. Subtract all these forced copies as well. For each color $a$, let $h_a$ be the number of residual edges whose fixed endpoint has color $a$; their other endpoints are uncolored leaves.

We again use a compressed flow. For every remaining stick type $\{a,b\}$ of positive multiplicity $q_{ab}$, create a type node with an incoming arc of capacity $q_{ab}$ from the source and outgoing arcs of capacity $q_{ab}$ to color nodes $a$ and $b$ (one arc for a loop). Give color node $a$ an arc to the sink of capacity $h_a$. A flow equal to the number of remaining sticks exists exactly when all residual edges can be supplied: each unit routed to color $a$ assigns one stick copy to an uncolored leaf edge whose fixed endpoint has color $a$.

The network has $O(m+c)$ arcs and polynomially bounded integral capacities, so one body color is tested in $(m+c)^{1+o(1)}$ time using the deterministic exact max-flow algorithm of van den Brand et al.~\cite{BrandEtAl2023Flow}. Trying at most $c$ body colors gives $c(m+c)^{1+o(1)}$ total time. \qed
\end{proof}

\subsection{Precolored Leaves}

\begin{theorem}\label{thm:length2-leaves}
Let every spider leg have length exactly two and every leaf be precolored. Internal vertices and the body may be precolored or uncolored. With $m$ host edges, \FMMP{} is solvable in $c(m+c)^{1+o(1)}$ time.
\end{theorem}

\begin{proof}
Again try every body color $s$ compatible with the body precolor. Fix one choice. A leg whose internal vertex is already precolored has both of its edge types forced: if its internal color is $a$ and its leaf color is $b$, it consumes one stick $\{s,a\}$ and one stick $\{a,b\}$. Subtract these copies, rejecting if either multiplicity is unavailable. We are left with $L$ legs whose internal vertices are uncolored and whose leaves remain precolored.

For each color $a$, let $n_a$ be the number of remaining leaves precolored $a$, let $D_a$ be the number of endpoints of color $a$ among the remaining sticks, and let $r_a$ be the number of remaining internal vertices that will receive color $a$. Every such internal vertex has degree two, every remaining leaf has degree one, and the body is incident with all $L$ remaining body edges. Therefore endpoint counting gives $D_a=2r_a+n_a$ for $a\ne s$, while $D_s=L+2r_s+n_s$. These equations uniquely determine every $r_a$. Reject if some $r_a$ is negative or nonintegral or if $\sum_ar_a\ne L$.

If $r_a$ internal vertices have color $a$, then exactly $r_a$ remaining body edges have type $\{s,a\}$. Reserve that many sticks of each type $\{s,a\}$, with type $\{s,s\}$ treated in the same way. Reject if any required multiplicity is unavailable. After the reservation, exactly $L$ sticks remain, one for each outer edge.

It remains to decide which internal color is paired with each precolored leaf. Regard every remaining stick $\{a,b\}$ as an undirected edge of a multigraph on the color set. Placing this stick on an outer edge whose internal vertex has color $a$ and whose leaf has color $b$ corresponds to orienting it from $a$ to $b$. We therefore need an orientation in which color $a$ is the tail of exactly $r_a$ edges and the head of exactly $n_a$ edges. A loop $\{a,a\}$ necessarily contributes one to both quantities.

This orientation can be found by a compressed flow. First remove all loops; if there are $\ell_a$ loops at color $a$, set the remaining tail demand to $r'_a=r_a-\ell_a$ and reject if $r'_a<0$. For every nonloop type $\{a,b\}$ of positive multiplicity $q_{ab}$, create one type node. Add an arc from the source to this node of capacity $q_{ab}$ and arcs from it to color nodes $a$ and $b$, again of capacity $q_{ab}$. Finally, give color node $a$ an arc to the sink of capacity $r'_a$. A flow equal to the number of nonloop remaining sticks chooses the tail of every stick copy and realizes exactly the required outdegrees. The required indegrees then follow automatically from the endpoint identity: after loops are removed, the number of nonloop stick endpoints at color $a$ equals $r'_a+(n_a-\ell_a)$.

From a feasible orientation, assign the edges directed into color $b$ bijectively to the $n_b$ leaves precolored $b$; their tails give the colors of the corresponding internal vertices. Conversely, every feasible puzzle solution induces exactly such an orientation and flow.

There are at most $m$ positive remaining stick types, so for one body color the flow network has $O(m+c)$ arcs and vertices, and all capacities are at most $m$. By the deterministic almost-linear-time exact max-flow algorithm of van den Brand et al.~\cite{BrandEtAl2023Flow}, the flow is found in $(m+c)^{1+o(1)}$ time. Trying at most $c$ body colors gives $c(m+c)^{1+o(1)}$ total time. \qed
\end{proof}

\section{Conclusion}

We established structural tractability and hardness boundaries for Matching-Match on dense, sparse, connected, and disconnected host graphs. In particular, uncolored paths and cycles have linear-time feasibility tests but $\#P$-complete counting versions, while counting remains polynomial on stars and complete graphs. Dense host graphs exhibit sharp thresholds in complement degree and uniform multipartite part size, whereas disconnected sparse graphs become hard already under very small local structure.

Two questions remain especially natural. Is \FMMP{} fixed-parameter tractable or W[1]-hard when parameterized by the number $k$ of parts of a complete $k$-partite host graph? For spiders whose legs have length at most two, is \FMMP{} polynomial-time solvable or NP-complete under arbitrary precoloring when the number of colors is unrestricted?

\bibliographystyle{splncs04}
\bibliography{references}

\end{document}